\documentclass[submission,copyright,creativecommons]{eptcs}

\usepackage{xspace}
\usepackage{tikz}
\usepackage{amsthm}
\usepackage{amssymb}
\usepackage{mathtools}
\usepackage{caption}
\usepackage{subcaption}
\usepackage[normalem]{ulem}
\usepackage{eurosym}

\providecommand{\event}{FMSPLE 2015} 
\newcommand{\eg}{\mbox{e.\,g.,}\xspace}
\newcommand{\etal}{\mbox{et\,al.}\xspace}
\newcommand{\ie}{\mbox{i.\,e.,}\xspace}
\newcommand{\cf}{\mbox{cf.}\xspace}
\newcommand{\ioco}{\textbf{ioco}\xspace}
\newcommand{\mioco}{\ensuremath{\miocomia}\xspace}
\newcommand{\mior}{\textbf{mior}\xspace}
\newcommand{\miorpre}{\textbf{mior}$_\leq$\xspace}

\newcommand{\mustarrow}{\longrightarrow_\Box}
\newcommand{\mayarrow}{\longrightarrow_\Diamond}

\newcommand{\oversetmust}[1]{\overset{#1}{\longrightarrow}_\Box}
\newcommand{\Oversetmust}[1]{\overset{#1}{\longrightarrow}_\Box}
\newcommand{\oversetmay}[1]{\overset{#1}{\longrightarrow}_\Diamond}
\newcommand{\Oversetmay}[1]{\overset{#1}{\longrightarrow}_\Diamond}
\newcommand{\transrel}[1]{\overset{#1}{\longrightarrow}}

\newtheorem{definition}{Definition}
\newtheorem{proposition}{Proposition}
\newtheorem{theorem}{Theorem}
\newtheorem{lemma}{Lemma}

\renewenvironment{proof}{\noindent{\itshape Proof.}}{\hfill\qed\smallskip

}

\newcommand\mayinit{\ensuremath{\textit{init}_\Diamond}\xspace}
\newcommand\mustinit{\ensuremath{\textit{init}_\Box}\xspace}
\newcommand\mayafter{\ensuremath{\,\after_\Diamond\,}\xspace}
\newcommand\mustafter{\ensuremath{\,\after_\Box\,}\xspace}
\newcommand\mayout{\ensuremath{\textit{Out}_\Diamond}\xspace}
\newcommand\mustout{\ensuremath{\textit{Out}_\Box}\xspace}
\newcommand\maystraces{\ensuremath{\textit{Straces}_\Diamond}\xspace}
\newcommand\muststraces{\ensuremath{\textit{Straces}_\Box}\xspace}

\newcommand\straces{\ensuremath{\textit{Straces}}\xspace}
\DeclareMathOperator{\after}{\,\mathbf{after}\xspace}
\newcommand\out{\ensuremath{\textit{Out}}\xspace}
\DeclareMathOperator{\miocomia}{\mathbf{mioco}_{\mathsf{MIA}}}
\DeclareMathOperator{\mathioco}{\mathbf{ioco}}
\DeclareMathOperator{\mathmior}{\mathbf{mior}}

\newcommand\miaref{\ensuremath{\,\sqsubseteq_{\mathsf{MIA}}\,}\xspace}
\newcommand\confref{\ensuremath{\,\sqsubseteq_{\mathsf{var}}\,}\xspace}
\newcommand\famlts[1]{\ensuremath{#1_{\mathit{fam}}}\xspace}

\usetikzlibrary{arrows,positioning}
\tikzset{
  every node/.style={},
  state/.style={circle,draw,fill=black,inner sep=.2em},
  emptystate/.style={circle,inner sep=.2em},
  labeledstate/.style={circle,draw,inner sep=.2em}, 
  may/.style={->,>=stealth',dashed},
  must/.style={->,>=stealth'},
  emptyarrow/.style={->,>=stealth',transparent}
}

\title{Towards an I/O Conformance Testing Theory for Software Product Lines based on Modal Interface Automata}
\author{Lars Luthmann$^{*}$ \qquad\qquad Stephan Mennicke\thanks{This work was partially supported by the DFG (German Research Foundation), grant GO-671/6-2.}
    \institute{Institute for Programming and Reactive Systems\\
    TU Braunschweig, Germany}
    \email{l.luthmann@tu-bs.de \quad\qquad mennicke@ips.cs.tu-bs.de}
    \and
    Malte Lochau\thanks{This work was partially supported by the DFG (German Research Foundation) under the Priority Programme SPP1593: Design For Future –-- Managed Software Evolution.}
    \institute{Realtime Systems Lab\\
    TU Darmstadt, Germany}
    \email{malte.lochau@es.tu-darmstadt.de}
}
\def\titlerunning{I/O Conformance Testing for SPL based on MIA}
\def\authorrunning{L.\ Luthmann, S.\ Mennicke \& M.\ Lochau}

\begin{document}
    \maketitle

%
\begin{abstract}
We present an adaptation of
input/output conformance (ioco)
testing principles to families
of similar implementation variants as
appearing in product line engineering.
Our proposed product line testing theory
relies on Modal Interface Automata (MIA)
as behavioral specification formalism.
MIA enrich I/O-labeled transition systems with
may/must modalities to distinguish
mandatory from optional behavior, thus
providing a semantic notion of intrinsic behavioral variability.
In particular, MIA constitute a restricted,
yet fully expressive subclass of I/O-labeled
modal transition systems, guaranteeing 
desirable refinement and compositionality properties.
The resulting modal-ioco relation
defined on MIA is preserved
under MIA refinement, which serves as 
variant derivation mechanism in our product line testing theory.
As a result, modal-ioco is proven correct in the 
sense that it coincides with traditional ioco to hold for 
every derivable implementation variant.
Based on this result, a family-based
product line conformance testing framework 
can be established.
\end{abstract}
%
\section{Introduction}\label{sec:introduction}
Modal transition systems (MTS) constitute an extension
to (labeled) transition systems (LTS) by enriching
the transition relation with a
\emph{may}/\emph{must} dichotomy~\cite{Larsen1988,larsen:modalspec}.
This way, behavioral system specifications based
on MTS leave open implementation freedom
by distinguishing \emph{mandatory} from \emph{optional}
behaviors, thus imposing a rigorous notion of (semantic)
\emph{refinement}~\cite{Alur1998}.
Considering Input/Output-labeled MTS in particular,
they provide a suitable
foundation for interface specifications
of component-based systems~\cite{Raclet2009}.
MTS incorporate a natural notion of interface \emph{compatibility}
and, thereupon, enjoy desirable compositionality
properties, being imperative for a comprehensive
\emph{interface theory}~\cite{raclet:modalinterfacetheory,Bauer2011}.

Based on the work of Fischbein et al.~in~\cite{Fischbein2006},
Larsen et al.~propose in~\cite{larsen:modalioautomata} to use modal specifications
as a basis for a \emph{behavioral variability theory}
for software product lines~\cite{Clements2001}.
A product line, therefore,
comprises a family of well-defined
\emph{implementation variants} derivable
from modal specifications using modal refinement, where
the validity of a variant is further restricted
due to its compatibility with other components and/or
a given environmental specification.
Based on this compact representation
of families of implementation variants,
a verification theory for product lines
has been developed in~\cite{Asirelli2011a}, combining
MTS with deontic logics to further restrict variable behaviors.
Those approaches allow for model-checking temporal properties
on entire families of implementation variants without explicitly
considering every particular variant, which is referred to as
\emph{family-based} product line analysis~\cite{Thum2014}.

However, besides those appealing family-based product line
verification approaches, the applicability of
modal specifications as a formal foundation
for a family-based product line testing theory
has not been intensively considered so far.
In particular, the input/output conformance
testing theory, initially introduced by Tretmans in~\cite{Tretmans1996},
is one of the most established formal frameworks to
reason about fundamental properties
of (model-based) functional testing approaches.
For this purpose, the \textbf{ioco} relation
imposes a notion of observational equivalence~\cite{Nicola1987}
between a test model specification, given
as an I/O-labeled transition system, and a (black box)
implementation under test, requiring the implementation
behaviors to conform to the specified behaviors.

To the best of our knowledge, there currently exist two approaches adapting I/O-conformance testing principles to product lines.
In Beohar and Mousavi~\cite{Beohar2014}, featured transition systems (FTS), initially proposed by Classen et al.~\cite{Classen2010}, are equipped with
I/O labels to enrich the \textbf{ioco} relation with explicit feature constraints.
This way, a family-based I/O-conformance testing framework
can be established, based on constraint-solving
capabilities as used in~\cite{Classen2011} for product line model checking.
In contrast, the approach proposed
in~\cite{lochau:modelbasedtesting}, called \textbf{mioco}, adapts the
key concepts of \textbf{ioco} to modal product line
specifications where I/O-labeled MTS (IOMTS) are used as specifications
of product lines under test.
A corresponding family-based I/O-conformance testing theory can
be built upon the notions of modal refinement and composition~\cite{larsen:modalioautomata}.
In this paper, we present an improved
elaboration of this initial approach to serve as a sound basis for
family-based product line conformance testing.
In particular, we make the
following contributions.
\begin{itemize}
	\item We consider a novel class of I/O-labeled modal
	transition systems, \ie Modal Interface Automata (MIA)~\cite{luettgen:modalinterfaceautomata},
	instead of IOMTS.
	MIAs slightly restrict IOMTS to guarantee desirable
	refinement and compositionality properties.
	\item We define the conformance relation $\miocomia$
	to relate product line implementations to product line specifications
	both given as MIAs.
	Thereupon, we clarify the assumptions to hold for specification
	and implementation in the spirit of classical \textbf{ioco}, \eg
	concerning input-enabledness and different concepts for input completions.
	One major challenge is to guarantee a proper treatment
	of the two kinds of implementation freedom apparent in product line
	specifications, namely \emph{variable} and \emph{unspecified} behaviors.
	\item In addition to the basic result in~\cite{lochau:modelbasedtesting} ensuring preservation of $\textbf{mioco}$ on IOMTS under refinement,
	we obtain strong results for our novel conformance relation $\miocomia$, concerning	soundness and completeness.
	Therefore, $\miocomia$ reflects the essence of family-based product line analysis by means of I/O-conformance testing~\cite{Thum2014}.
\end{itemize}
Here, we focus our considerations
to testing single components, and, therefore, also
omitting $\tau$ transitions within modal specifications.
However, the results obtained in this paper pave
the way to a compositional and family-based product line testing theory.

%
The remainder of this paper
is structured as follows.
Sect.~\ref{sec:preliminaries} provides
a brief repetition of
input/output conformance defined on I/O-labeled
transition systems, as well as
the foundations of modal transition systems.
In Sect.~\ref{sec:model}, we introduce MIA
as a new model for product line specifications
and describe variant derivation semantics in terms
of MIA refinement.
In Sect.~\ref{sec:testing_theory}, we propose
an adaptation of input/output conformance notions to
MIA and define approaches for achieving
input-enabledness via completions.
Our main result concerning the correctness
of modal-ioco on MIA are formulated and proven
in Sect.~\ref{sec:results}.
Sect.~\ref{sec:discussion} concludes with
an outlook on our ongoing and future research directions.

\section{Preliminaries}\label{sec:preliminaries}
We start with an overview on I/O-labeled transition systems and I/O conformance testing.
Furthermore, we present modal transition systems (MTS) laying the
foundation for modal interface automata.

Labeled transition systems (LTS) constitute a well-established model for discrete-state reactive system behaviors.
The behavior of an LTS is specified by means of a labeled transition relation $\longrightarrow\subseteq Q \times \textit{act} \times Q$ on a set of states $Q$ and an alphabet of actions $\textit{act}$.
To serve as a test model specification for input/output conformance testing, the subclass of \emph{I/O labeled transition systems} is considered, dividing the set $\textit{act}$
into disjoint subsets of controllable \emph{input}
actions $I$ and observable \emph{output} actions $O$.
In addition, \emph{internal} actions are usually summarized
under the special symbol $\tau\not\in(I\cup O)$.
However, we do not consider $\tau$ transitions in this paper.
\begin{definition}[I/O Labeled Transition System]\label{def:lts}
	An \emph{I/O labeled transition system} (IOLTS) is a tuple $(Q,I,O,\longrightarrow)$, where
	\begin{itemize}
		\item $Q$ is a countable set of states,
		\item $I$ and $O$ are disjoint sets of input actions and output actions, respectively, and
		\item $\longrightarrow\subseteq Q \times (I\cup O) \times Q$ is a labeled transition relation.
	\end{itemize}
\end{definition}
Note that (IO)LTS usually do not comprise a predefined initial state,
as it is either identified with its entire set of states $Q$, or some state $q\in Q$ denoting the initial state.
By $q\overset{a}{\longrightarrow}q'$ we mean that $(q,a,q')\in\longrightarrow$ holds, and
we write $q\overset{a}{\longrightarrow}$ as a short hand for $\exists q'\in Q : q\overset{a}{\longrightarrow} q'$.
We further denote a {\em path} $q_0\overset{a_1}{\longrightarrow}q_1\overset{a_2}{\longrightarrow}\cdots\overset{a_{n-1}}{\longrightarrow}q_{n-1}\overset{a_n}{\longrightarrow}q_n$ by $q_0\overset{\sigma}{\longrightarrow}q_n$, where $\sigma = a_1 a_2 \ldots a_n \in (I\cup O)^{*}$ is called a {\em trace}.

In the input/output conformance relation (\cf Sect.~\ref{subsec:ioco}), an implementation, represented as an I/O-labeled transition system, is assumed to be \emph{input-enabled}, \ie to never reject any inputs.
This yields the subclass of \emph{I/O transition systems}.
\begin{definition}[I/O Transition System]\label{def:iots}
	A state $q\in Q$ of an IOLTS $Q$ is \emph{input-enabled} iff for all $i\in I$, there exists a state $q'\in Q$ such that $q\transrel{i} q'$.
	$Q$ is an {\em I/O Transition System} (IOTS) iff all $q\in Q$ are input-enabled.
\end{definition}
In deviation from Tretmans~\cite{Tretmans1996}, we employ strong input-enabledness, as we do not consider internal behavior.
Figure~\ref{fig:lts_example} shows three sample LTS, modeling different
variants of a vending machine.
All of these vending machines have in common that they accept money in the initial state (the topmost state) and are capable of dispensing tea.
However, some of them may also dispense coffee and notify the user about errors.
Transitions are labeled with either input labels (prefix $?$), or output labels (prefix $!$).
%
The LTS in Figure~\ref{fig:lts_example_b} accepts 1\euro{} or 2\euro{} coins from customers.
After inserting 2\euro{}, change is returned and the customer is allowed to
choose coffee or tea, or to refill cups.
Next, the vending machine dispenses the selected beverage.
The LTS in Figure~\ref{fig:lts_example_a} is an IOTS, as every
state accepts all possible inputs, \ie \emph{1\euro}, \emph{2\euro}, \emph{cups}, and \emph{tea}.
Label \emph{I} denotes that a transition exists
for each input symbol, unless a state already accepts an input.
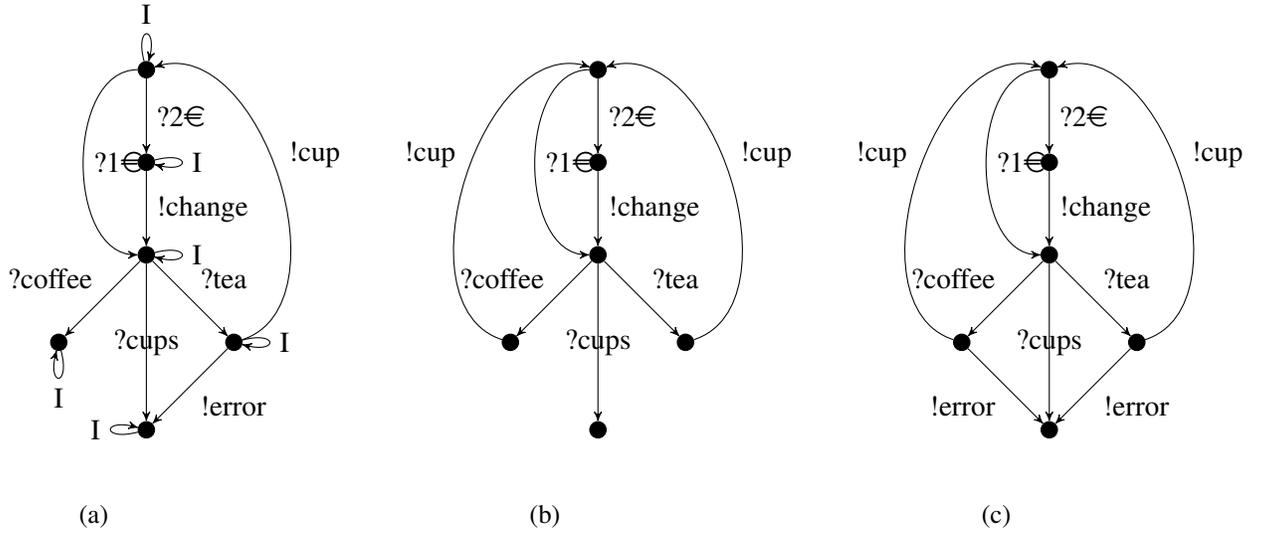
\begin{figure}
	\begin{subfigure}[b]{0.25\textwidth}
		\centering

\begin{tikzpicture}

\node[state] (q0) {};
\node[state, below=of q0] (q1) {};
\node[state, below=of q1] (q2) {};
\node[state, below left=of q2] (q3) {};
\node[state, below right=of q2] (q4) {};
\node[state, below right=of q3] (q5) {};

\draw[must] (q0) to node[auto] {?2\euro} (q1);
\draw[must,bend right=90] (q0) to node[auto] {?1\euro} (q2);
\draw[must] (q1) to node[auto] {!change} (q2);
\draw[must] (q2) to node[auto,swap] {?coffee} (q3);
\draw[must] (q2) to node[auto] {?tea} (q4);
\draw[must] (q2) to node[pos=.5] {?cups} (q5);
\draw[emptyarrow] (q3) to node[auto,swap] {\phantom{!error}} (q5);
\draw[must] (q4) to node[auto] {!error} (q5);
\draw[emptyarrow,bend left=90] (q3) to node[auto] {} (q0);
\draw[must,bend right=90] (q4) to node[auto,swap] {!cup} (q0);
\draw[must,loop above] (q0) to node[auto] {I} (q0);
\draw[must,loop right] (q1) to node[auto] {I} (q1);
\draw[must,loop right] (q2) to node[auto] {I} (q2);
\draw[must,loop right] (q4) to node[auto] {I} (q4);
\draw[must,loop left] (q5) to node[auto] {I} (q5);
\draw[must,loop below] (q3) to node[auto] {I} (q3);

\end{tikzpicture}
		\subcaption{}
		\label{fig:lts_example_a}
	\end{subfigure}
	\hfill
	\begin{subfigure}[b]{0.25\textwidth}
		\centering

\begin{tikzpicture}

\node[state] (q0) {};
\node[state, below=of q0] (q1) {};
\node[state, below=of q1] (q2) {};
\node[state, below left=of q2] (q3) {};
\node[state, below right=of q2] (q4) {};
\node[state, below right=of q3] (q5) {};

\draw[must] (q0) to node[auto] {?2\euro} (q1);
\draw[must,bend right=90] (q0) to node[auto] {?1\euro} (q2);
\draw[must] (q1) to node[auto] {!change} (q2);
\draw[must] (q2) to node[auto,swap] {?coffee} (q3);
\draw[must] (q2) to node[auto] {?tea} (q4);
\draw[must] (q2) to node[pos=.5] {?cups} (q5);
\draw[emptyarrow] (q3) to node[auto,swap] {\phantom{!error}} (q5);
\draw[emptyarrow] (q4) to node[auto] {\phantom{!error}} (q5);
\draw[must,bend left=90] (q3) to node[auto] {!cup} (q0);
\draw[must,bend right=90] (q4) to node[auto,swap] {!cup} (q0);
\draw[emptyarrow,loop above] (q0) to node[auto] {I} (q0);
\draw[emptyarrow,loop right] (q1) to node[auto] {I} (q1);
\draw[emptyarrow,loop right] (q2) to node[auto] {I} (q2);
\draw[emptyarrow,loop right] (q4) to node[auto] {I} (q4);
\draw[emptyarrow,loop left] (q5) to node[auto] {I} (q5);

\end{tikzpicture}
		\subcaption{}
		\label{fig:lts_example_b}
	\end{subfigure}
	\hfill
	\begin{subfigure}[b]{0.25\textwidth}
		\centering

\begin{tikzpicture}

\node[state] (q0) {};
\node[state, below=of q0] (q1) {};
\node[state, below=of q1] (q2) {};
\node[state, below left=of q2] (q3) {};
\node[state, below right=of q2] (q4) {};
\node[state, below right=of q3] (q5) {};

\draw[must] (q0) to node[auto] {?2\euro} (q1);
\draw[must,bend right=90] (q0) to node[auto] {?1\euro} (q2);
\draw[must] (q1) to node[auto] {!change} (q2);
\draw[must] (q2) to node[auto,swap] {?coffee} (q3);
\draw[must] (q2) to node[auto] {?tea} (q4);
\draw[must] (q2) to node[pos=.5] {?cups} (q5);
\draw[must] (q3) to node[auto,swap] {!error} (q5);
\draw[must] (q4) to node[auto] {!error} (q5);
\draw[must,bend left=90] (q3) to node[auto] {!cup} (q0);
\draw[must,bend right=90] (q4) to node[auto,swap] {!cup} (q0);
\draw[emptyarrow,loop above] (q0) to node[auto] {I} (q0);
\draw[emptyarrow,loop right] (q1) to node[auto] {I} (q1);
\draw[emptyarrow,loop right] (q2) to node[auto] {I} (q2);
\draw[emptyarrow,loop right] (q4) to node[auto] {I} (q4);
\draw[emptyarrow,loop left] (q5) to node[auto] {I} (q5);

\end{tikzpicture}
		\subcaption{}
		\label{fig:lts_example_c}
	\end{subfigure}
	\caption{Sample LTS of a simple vending machine, adapted from~\cite{lochau:modelbasedtesting}.}
	\label{fig:lts_example}
\end{figure}
\subsection{Input/Output Conformance}\label{subsec:ioco}
An implementation $i$, given as an IOTS, I/O-conforms
to a specification $s$, given as an IOLTS, if all observable output symbols of $i$ after any
possible input sequence $\sigma$ of $s$ are permitted by $s$.
That means that a system specification states the allowed output behavior.
For this to hold, the set $\textit{Out}(P)$ of output actions enabled in any possible state $p\in P$ of $i$ reachable via a sequence $\sigma$, denoted by $P=i\after\sigma$, must be included in the corresponding set $\textit{Out}(Q)$ with $Q=s\after\sigma$.
To rule out trivial implementations never showing any outputs, the concept of \emph{quiescence} is introduced by means of an observable action $\delta$ to explicitly permit the absence (suspension) of any output in a state.
The definitions of this section follow~\cite{Tretmans1996}.
\begin{definition}
	Let $Q$ be an IOLTS, $p\in Q$, $P\subseteq Q$, and $\sigma\in(I\cup O\cup\{\delta\})^*$.
	\begin{itemize}
		\item $init(p):=\{\mu\in(I\cup O)\mid p\overset{\mu}{\longrightarrow}\}$,
		\item $p$ is \emph{quiescent}, denoted by $\delta(p)$, iff $init(p)\subseteq I$,
		\item $p\after\sigma:=\{q\in Q\mid p\overset{\sigma}{\longrightarrow}q\}$,
		\item $Out(P):=\{\mu\in O\mid \exists p\in P:p\overset{\mu}{\longrightarrow}\}\cup\{\delta\mid \exists p\in P:\delta(p)\}$, and
		\item $\straces(p):=\{\sigma\in(I\cup O\cup\{\delta\})^*\mid p\overset{\sigma}{\longrightarrow}\}$, where $q\overset{\delta}{\longrightarrow}q$ iff $\delta(q)$.
	\end{itemize}
\end{definition}
I/O conformance requires any reaction of an implementation $i$ to every possible environmental behavior $\sigma$ to be checked against those of its specification $s$, even if no proper reaction for $\sigma$ is actually specified by $s$.
Hence, conformance testing is usually limited to positive testing, \ie only considering behaviors being explicitly specified in $s$, \ie contained in the {\em suspension traces} ($\straces(s)$) of $s$.
\begin{definition}[Input/Output Conformance]\label{def:ioco}
	Let $s$ be an IOLTS and $i$ an IOTS with the same sets of inputs and outputs.
	$i\mathioco s:\Leftrightarrow\forall\sigma\in Straces(s):Out(i\after\sigma)\subseteq Out(s\after\sigma)$.
\end{definition}
Assuming the IOLTS of Figure~\ref{fig:lts_example_c} to be a specification $s$ and the IOTS of Figure~\ref{fig:lts_example_a} to be an implementation $i$, then $i\mathioco s$ holds, as $i$ does not show any unspecified output behavior.
However, considering the IOLTS of Figure~\ref{fig:lts_example_b} as a specification $s$ for $i$, then $i\mathioco s$ does not hold as $i$ exhibits output behavior \emph{error}, violating conformance of $i$ to $s$.

The \ioco relation permits implementation freedom as only one specified output behavior must be implemented.
In addition, if there are unspecified inputs for state $q$ in the specification $s$, then an implementation may react with arbitrary outputs to those unspecified inputs,
as those behaviors do not occur in the suspension traces of $s$ and are, therefore, never tested.
However, for product lines, we further need the possibility to (1) explicitly express variability within specifications and, therefore, to (2) distinguish mandatory from optional behavior, which leads us to Modal Transition Systems (MTS).
\subsection{Modal Transition Systems}\label{subsec:mts}
%
To specify behavioral variability of product lines, we use Modal Transition Systems (MTS) according to Larsen~\cite{larsen:modalspec,larsen:modalioautomata} as a basis.
MTS are based on LTS but distinguish between so called must and may transitions, specifying mandatory behavior as well as allowed behavior, respectively.
By definition any must transition is a may transition as any mandatory behavior must also be allowed.
Therefore, only may transitions not underlying must transitions denote optional behavior.
Accordingly, we call may transitions that are not must transitions \emph{optional} and must transitions \emph{mandatory}.
Additionally, absent transitions denote forbidden behavior.
\begin{definition}[Modal Transition System]\label{def:mts}
	A tuple $Q=(Q,A,\mustarrow,\mayarrow)$ is a \emph{Modal Transition System} (MTS) iff
	\begin{itemize}
		\item $Q$ is a finite set of states,
		\item $A$ is a set of actions,
		\item $\mustarrow\subseteq Q\times A\times Q$ is the labeled \emph{must transition} relation,
		\item $\mayarrow\subseteq Q\times A\times Q$ is the labeled \emph{may transition} relation, and
	\end{itemize}
	$Q$ is \emph{syntactically consistent}, \ie $q\oversetmust{a}q'$ implies $q\oversetmay{a}q'$.
\end{definition}
Note that, in our setting, we assign $I\cup O$ to $A$ which defines {\em I/O-labeled MTS}.
MTS allow us to superimpose several systems into one larger system, from which the original systems are derivable via modal refinement.
Therefore, modal {\em refinement} preserves mandatory and forbidden behavior, whereas optional behavior may turn into either mandatory or forbidden behavior.
Figure~\ref{fig:mts_example} shows a sample MTS.
Therein, solid edges depict mandatory behavior and dashed edges depict optional behavior.
The MTS in Figure~\ref{fig:mts_example} combines all IOLTS from Figure~\ref{fig:lts_example} into one.
This is achieved by making the behaviors being common to all IOLTS variants mandatory behaviors, whereas the variable behaviors become optional in the MTS.
%
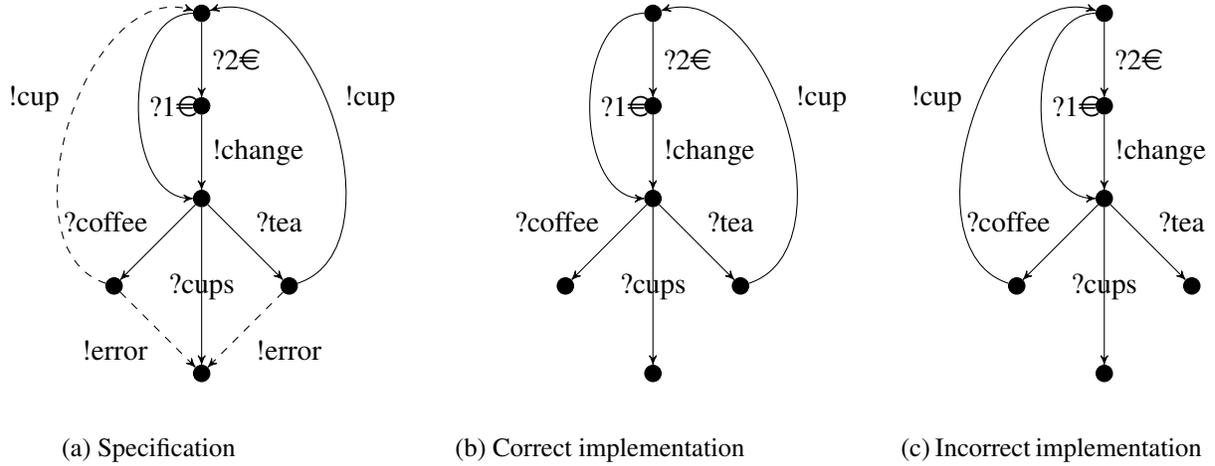
\begin{figure}
	\begin{subfigure}[b]{0.25\textwidth}
		\centering

\begin{tikzpicture}

\node[state] (q0) {};
\node[state, below=of q0] (q1) {};
\node[state, below=of q1] (q2) {};
\node[state, below left=of q2] (q3) {};
\node[state, below right=of q2] (q4) {};
\node[state, below right=of q3] (q5) {};

\draw[must] (q0) to node[auto] {?2\euro} (q1);
\draw[must,bend right=90] (q0) to node[auto] {?1\euro} (q2);
\draw[must] (q1) to node[auto] {!change} (q2);
\draw[must] (q2) to node[auto,swap] {?coffee} (q3);
\draw[must] (q2) to node[auto] {?tea} (q4);
\draw[must] (q2) to node[pos=.5] {?cups} (q5);
\draw[may] (q3) to node[auto,swap] {!error} (q5);
\draw[may] (q4) to node[auto] {!error} (q5);
\draw[may,bend left=90] (q3) to node[auto] {!cup} (q0);
\draw[must,bend right=90] (q4) to node[auto,swap] {!cup} (q0);

\end{tikzpicture}
		\subcaption{Specification}
		\label{fig:mts_example}
	\end{subfigure}
	\hfill
	\begin{subfigure}[b]{0.25\textwidth}
		\centering

\begin{tikzpicture}

\node[state] (q0) {};
\node[state, below=of q0] (q1) {};
\node[state, below=of q1] (q2) {};
\node[state, below left=of q2] (q3) {};
\node[state, below right=of q2] (q4) {};
\node[state, below right=of q3] (q5) {};

\draw[must] (q0) to node[auto] {?2\euro} (q1);
\draw[must,bend right=90] (q0) to node[auto] {?1\euro} (q2);
\draw[must] (q1) to node[auto] {!change} (q2);
\draw[must] (q2) to node[auto,swap] {?coffee} (q3);
\draw[must] (q2) to node[auto] {?tea} (q4);
\draw[must] (q2) to node[pos=.5] {?cups} (q5);
\draw[emptyarrow] (q3) to node[auto,swap] {\phantom{!error}} (q5);
\draw[emptyarrow] (q4) to node[auto] {\phantom{!error}} (q5);
\draw[emptyarrow,bend left=90] (q3) to node[auto] {\phantom{!cup}} (q0);
\draw[must,bend right=90] (q4) to node[auto,swap] {!cup} (q0);

\end{tikzpicture}
		\subcaption{Correct implementation}
		\label{fig:mioco_example_correct}
	\end{subfigure}
	\hfill
	\begin{subfigure}[b]{0.25\textwidth}
		\centering

\begin{tikzpicture}

\node[state] (q0) {};
\node[state, below=of q0] (q1) {};
\node[state, below=of q1] (q2) {};
\node[state, below left=of q2] (q3) {};
\node[state, below right=of q2] (q4) {};
\node[state, below right=of q3] (q5) {};

\draw[must] (q0) to node[auto] {?2\euro} (q1);
\draw[must,bend right=90] (q0) to node[auto] {?1\euro} (q2);
\draw[must] (q1) to node[auto] {!change} (q2);
\draw[must] (q2) to node[auto,swap] {?coffee} (q3);
\draw[must] (q2) to node[auto] {?tea} (q4);
\draw[must] (q2) to node[pos=.5] {?cups} (q5);
\draw[emptyarrow] (q3) to node[auto,swap] {\phantom{!error}} (q5);
\draw[emptyarrow] (q4) to node[auto] {\phantom{!error}} (q5);
\draw[must,bend left=90] (q3) to node[auto] {!cup} (q0);
\draw[emptyarrow,bend right=90] (q4) to node[auto,swap] {\phantom{!cup}} (q0);

\end{tikzpicture}
		\subcaption{Incorrect implementation}
		\label{fig:mioco_example_incorrect}
	\end{subfigure}
	\caption{Figure~\ref{fig:mts_example} shows an MTS combining all systems from Figure~\ref{fig:lts_example}. Figures~\ref{fig:mioco_example_correct} and~\ref{fig:mioco_example_incorrect} show a correct and an incorrent implementation regarding \mioco(see Section~\ref{sec:testing_theory}).}
	\label{fig:mioco_example}
\end{figure}
Larsen \etal~\cite{larsen:modalioautomata} also defined input-enabledness for MTS and an according input-enabled MTS version to be called \emph{I/O modal transition systems} (IOMTS).
However, there are two flavors of input-enabledness: \emph{must-input-enabledness} and \emph{may-input-enabledness}, both being defined as canonical extensions of Def.~\ref{def:iots}.
A may-input-enabled MTS is called {\em I/O-labeled MTS} (IOMTS).
%
%
For an overview on IOMTS and the according I/O-conformance relation, we refer to~\cite{lochau:modelbasedtesting}.
Another option for adding modalities to system specifications are {\em Modal Interface Automata}, being a restricted subclass of IOMTS.

\section{Modal Interface Automata}\label{sec:model}
%
The model we employ in this paper is called {\em Modal Interface Automata} (MIA) which are basically input-deterministic I/O-labeled MTS.
Input-determinism is a desirable property in model-based testing, as it makes testing procedures manageable by ensuring that some inputs are not infinitely often neglected during test scenarios, as imposed by non-deterministic inputs.
Furthermore, specified inputs are always mandatory, but unspecified inputs are implicitly allowed.
This restriction yields refinement and composition properties beneficial for both modeling product line specifications and implementations with behavioral variability, as well as modal I/O-conformance testing as described in Sect.~\ref{sec:testing_theory}.
The MTS depicted in Figure~\ref{fig:mioco_example} are in fact MIAs, as they exhibit input-determinism and every input transition is mandatory.
\begin{definition}[Modal Interface Automaton]\label{def:mia}
  A tuple $Q=(Q,I,O,\mustarrow,\mayarrow)$ is a \emph{Modal Interface Automaton} (MIA), where $(Q,I\cup O, \mustarrow,\mayarrow)$ is an MTS with disjoint alphabets \emph{I}, \emph{O} and for all $i\in I$:
  \begin{itemize}
    \item $q\oversetmust{i}q'$ and $q\oversetmust{i}q''$ implies $q'=q''$ (\ie we require input-determinism),
    \item $q\oversetmay{i}q'$ implies $q\oversetmust{i}q'$ (\ie all inputs are mandatory behavior).
  \end{itemize}
\end{definition}
In deviation to L\"uttgen and Vogler~\cite{luettgen:modalinterfaceautomata}, we do not employ \emph{disjunctive} MTS, as they are not needed for our purposes.
Furthermore, we limit our considerations to MIAs without internal behaviors, \ie $\tau$ transitions, which is no limitation, as all our results remain valid for MIAs with internal behavior.
For future work, we plan to investigate our testing theory under parallel composition for which a treatment of internal transitions is inevitable.
L\"uttgen and Vogler define an operator for parallel composition of MIA similar to interface automata~\cite{alfaro:interfacebased}.
They identify \emph{error states} arising from the composition of incompatible states, and remove them, as well as all states from which reaching some error state is no more preventable by environmental inputs.
This is similar to the operator by Larsen et al.~\cite{larsen:modalioautomata}, but, in contrast to IOMTS, composability of MIA is based on the compatible component semantics rather than syntactic criteria.

Each input transition of a MIA is, by definition, mandatory.
However, this does not limit the expressiveness of MIA compared to IOMTS, as input transitions are always implicitly allowed by {\em modal refinement}.
Modal refinement is a crucial notion in modal theories, as they constitute an implementation relation that preserves mandatory behaviors, but also leaves implementation freedom concerning optional and unspecified behaviors.
Intuitively, a MIA $p$ refines $q$ if (1) the optional output behavior of $p$ is simulated by $q$, and (2) all mandatory behavior of $q$ is simulated by $p$, thus imposing an {\em alternating simulation}~\cite{luettgen:modalinterfaceautomata}.
%
\begin{definition}[MIA Refinement]\label{def:mia_refinement}
  Let $P, Q$ be MIAs over $I$ and $O$.
  A relation $\mathcal{R}\subseteq P\times Q$ is a {\em MIA-refinement} iff for all $(p, q) \in\mathcal{R}$:
  \begin{enumerate}
    \item $q\oversetmust{a} q'$ where $a\in I\cup O$ implies $\exists p' : p \oversetmust{a} p'$ and $(p',q')\in\mathcal{R}$,
    \item $p\oversetmay{\alpha} p'$ where $\alpha\in O$ implies $\exists q' : q\oversetmay{\alpha} q'$ and $(p',q')\in\mathcal{R}$.
  \end{enumerate}
  If there is a MIA-refinement $\mathcal{R}$ such that $(p, q)\in\mathcal{R}$, then {\em $p$ MIA-refines $q$}, denoted by $p\miaref q$.
\end{definition}
The most desirable property for composition operators in
modal system theory is their preservation of modal refinement.
L\"uttgen and Vogler show this to hold for parallel MIA composition, and also for {\em conjunction} and {\em disjunction} of MIAs~\cite{luettgen:modalinterfaceautomata}.
In contrast to MTS, not all unspecified transitions of MIAs refer to forbidden behavior, but only those being outputs.
Input transitions are always implicitly allowed and, therefore, in Def.~\ref{def:mia_refinement}, only optional outputs of the refined MIA must be simulated by the unrefined one.

In this paper, we interpret MIAs with mandatory and optional behaviors as families of similar system variants.
In this regard, the refinement notion serves {\em variant derivation} such that in $p\miaref q$, $p$ represents a variant of $q$, where a (partially) refined $p$ may still contain optional behavior.
Furthermore, as unspecified inputs are implicitly allowed under MIA-refinement, $p$ may also contain additional behaviors, which is not feasible for a product line specification.
In order to obtain a sound variant derivation mechanism, we require it to finally yield an IOLTS, which does not incorporate optional behavior.
Hence, this IOLTS variant $p$ refines a MIA $q$, but is restricted to behaviors allowed in $q$.
\begin{definition}[Variant Derivation]\label{def:conf_refinement}
  Let $(P,I,O,\transrel{})$ be an IOLTS and $Q$ be a MIA over $I$ and $O$.
  $p\in P$ is a {\em variant of $q\in Q$}, denoted by $p\confref q$, iff (1) $p\miaref q$, \ie there is a MIA-refinement $\mathcal{R}$ between $(P, I, O, \transrel{}, \transrel{})$ and $Q$ such that $(p,q)\in\mathcal{R}$ and (2) $\transrel{}\subseteq\mayarrow$.
\end{definition}
Thus, a variant derivation is a special kind of MIA-refinement, ensuring that every optional transition of the specification is either removed from, or definitely included in the variant.
There is a close relationship between traces of variants $p\confref q$ and may-traces of $q$.
\begin{lemma}\label{lemma:lts_to_modal}
  Let $P$ be an IOLTS with $p\in P$ and $Q$ a MIA with $q\in Q$.
  If $p\confref q$, then for each $w\in (I\cup O)^{*}$ with $p\transrel{w}$ it holds that $q\oversetmay{w}$.
\end{lemma}
\noindent Based on MIA refinement and MIA variant derivation mechanism, we define a modal version of \ioco.

\section{I/O Conformance Testing for MIAs}\label{sec:testing_theory}
In Sect.~\ref{subsec:ioco}, we have already discussed I/O conformance on IOLTS, allowing a certain degree of variability in implementations.
However, we propose to apply modal specifications to explicitly capture variability also within specifications, as being inherent to SPLs.
We now define the foundations of modal input/output conformance~\cite{lochau:modelbasedtesting} and introduce completion strategies for constructing input-enabled modal interface automata.
%
\subsection{Modal Input/Output Conformance}
Intuitively, a modal implementation conforms to a modal specification if it does not exceed the allowed outputs (may) and preserves all mandatory outputs (must).
We adapt the notion on I/O conformance to the MIA-framework, accordingly.
\begin{definition}\label{def:modal-init}
	Let $Q$ be a MIA over $I$ and $O$, $p\in Q$, $P\subseteq Q$, $\sigma\in(I\cup O \cup \{ \delta_\Box, \delta_\Diamond \})^*$, and $\gamma\in\{\Box,\Diamond\}$.
	\begin{itemize}
		\item $init_\gamma(p):=\{\mu\in(I\cup O)\mid p\overset{\mu}{\longrightarrow}_\gamma\}$,
		\item $p$ is \emph{may-quiescent}, denoted by $\delta_\Diamond(p)$, iff $\mustinit(p)\subseteq I$ and $p$ is \emph{must-quiescent}, denoted by $\delta_\Box(p)$, iff $\mayinit(p)\subseteq I$,
		\item $p\after_\gamma\sigma:=\{q\in Q \mid p\overset{\sigma}{\longrightarrow}_\gamma q\}$,
		\item $Out_\gamma(P):=\{\mu\in O\mid\exists p\in P:p\overset{\mu}{\longrightarrow}_\gamma\}\cup\{\delta_\gamma\mid\exists p\in P:\delta_\gamma(p)\}$, and
		\item $Straces_\gamma(p):=\{\sigma\in(I\cup O\cup\{\delta\})^*\mid p\Oversetmay{\sigma}\}$, where $q\overset{\delta}{\longrightarrow} q$ iff $\delta_\gamma(q)$.
	\end{itemize}
\end{definition}
Due to the property of syntactic consistency of MIAs, similar properties are induced for the notions of Def.~\ref{def:modal-init}.
For instance, $\mayinit(p)\subseteq I$ implies $\mustinit(p)\subseteq I$ as $a\in I$, $q\oversetmay{a} q'$ implies $q\oversetmust{a} q'$ and, therefore, must-quiescence of any state $p$ implies may-quiescence of $p$.
\begin{proposition}\label{lemma:may-must-preservation}
Let $Q$ be a MIA over $I$ and $O$, $p\in Q$, $P\subseteq Q$, and $\sigma\in (I\cup O\cup \{\delta_\Diamond, \delta_\Box \})^{*}$.
Then the following statements hold.
\begin{enumerate}
	\item $\mustinit(p)\subseteq\mayinit(p)$,
	\item if $\delta_\Box(p)$, then $\delta_\Diamond(p)$,
	\item $p \mustafter \sigma \subseteq p \mayafter \sigma$,
	\item $\mustout(P) \subseteq \mayout(P)$, and
	\item $\muststraces(p) \subseteq \maystraces(p)$.
\end{enumerate}
\end{proposition}
Proofs follow from Def.~\ref{def:mia} and Def.~\ref{def:modal-init}.
Due to the required input-determinism, MIAs may only show non-determinstic behavior due to conflicting output transitions.
Thus, when considering a trace $w$ of a state $q$, it holds that $|q\mayafter w| \geq 1$.
Considering MIA-refinement $p\miaref q$, each state in $p\mayafter w$ relates to some state in $q\mayafter w$.
Conversely, the same holds for states of $p\mustafter w$ and $q\mustafter w$.
\begin{lemma}\label{lemma:refinement-outstates}
	Let $p, q$ be MIAs such that $p\miaref q$.
	\begin{enumerate}
		\item $\forall\sigma\in (I\cup O \cup \{ \delta_\Diamond \})^{*} : q\mayafter\sigma \neq \emptyset \Rightarrow \forall p'\in p\mayafter\sigma : \exists q'\in q\mayafter\sigma : p'\miaref q'$
		\item $\forall\sigma\in (I\cup O \cup \{ \delta_\Box \})^{*} : p\mustafter\sigma \neq \emptyset \Rightarrow \forall q'\in q\mustafter\sigma : \exists p'\in p\mustafter\sigma : p'\miaref q'$
	\end{enumerate}
\end{lemma}
This property is very useful when arguing on paths of MIAs related under MIA-refinement (\cf~Theorem~\ref{theorem:compl}).
While for MTS we distinguished may from must input-enabledness, there is no difference between both in case of MIA, as inputs in MIAs are mandatory.
\begin{definition}[Input-Enabledness for MIA]\label{def:input_enabledness_mia}
	A MIA $Q$ is \emph{input-enabled} iff for all $q\in Q$ and for all $i\in I$, it holds that $q\oversetmust{i}$.
\end{definition}
Henceforth, we require product line implementations $i$ to be given as input-enabled MIAs in order to meet the assumptions originally made by \ioco that implementations do not deadlock while processing inputs not being serviced by the implementation.
Input-enabledness of MIA is preserved under MIA-refinement.
\begin{lemma}\label{lemma:input-enabledness}
	Let $q, r$ be MIAs over $I$ and $O$ such that $r\miaref q$.
	If $q$ is input-enabled, then $r$ is input-enabled.
\end{lemma}
This also holds for variant $p$ derived from $q$.
In \ioco, no distinction is made between specified mandatory and optional behavior.
A first conformance relation supporting optional behaviors is \mior~\cite{lochau:modelbasedtesting}.
It holds that $i\mathmior s$ in case of trace inclusion of may-suspension-traces as well as must-suspension-traces, respectively.
However, if we interpret the set of must-behaviors specified by $s$ as the product line core behavior incorporated by all variants, this notion of conformance fails to fully capture this intuition.
Suspension trace inclusion solely ensures \emph{some} behaviors of the specified behaviors to be actually implemented (if any), but it does not differentiate within the set of allowed behaviors between mandatory and optional ones.

Figure~\ref{fig:mioco_example} illustrates the weakness of \mior.
Assuming Figure~\ref{fig:mts_example} as specification and the other two Figures to be implementations, then both implementations are correct under \mior.
For Figure~\ref{fig:mioco_example_correct}, this is obvious as only optional behavior is left out.
The problem of \mior is depicted in the implementation of Figure~\ref{fig:mioco_example_incorrect}.
Therein, the mandatory behavior outputting the \emph{cup} after the input \emph{tea} is left out but the \mior relation still holds as no behavior is added.
This contradicts the intention of mandatory behaviors as core behaviors of all product line variants.
To overcome this drawback, consider an alternative definition for I/O conformance, denoted by \miorpre~\cite{lochau:modelbasedtesting}, being closer to the very essence of modal refinement requiring alternating suspension trace inclusions.
The \miorpre relation requires implementation $i$ to show at least all mandatory behaviors and at most the allowed behaviors of a specification $s$.
Following this idea, the respective modal version of \ioco is defined as follows.
\begin{definition}[Modal Input/Output Conformance]\label{def:mioco}
	Let $s, i$ be MIAs over $I, O$ and $i$ being input-enabled.
	$i\miocomia s$ iff
	\begin{enumerate}
		\item $\forall\sigma\in Straces_\Diamond(s):Out_\Diamond(i\after_\Diamond\sigma)\subseteq Out_\Diamond(s\after_\Diamond\sigma)$, and
		\item $\forall\sigma\in Straces_\Box(i):Out_\Box(s\after_\Box\sigma)\subseteq Out_\Box(i\after_\Box\sigma)$.
	\end{enumerate}
\end{definition}
In the first part of checking $i\miocomia s$, we consider all specified suspension-traces and essentially check $i\mathioco s$.
In the second part, we only consider must-suspension-traces of $i$ and essentially check $s\mathioco i$.
That way, we make sure that the implementation does not add forbidden behavior or ignores mandatory behavior.
Requiring input-enabledness for specifications of $i$ is infeasible for realistic test modeling approaches.
However, an artificial input-enabledness for incomplete specifications of $i$ may always be achieved by {\em completions} of $i$ (cf. Sect.~\ref{subsec:completions}).

Let us reconcile Figure~\ref{fig:mioco_example} with the \mioco relation instead of \mior.
Again, Figure~\ref{fig:mts_example} depicts the specification and the other two figures represent the implementations.
The implementation in Figure~\ref{fig:mioco_example_correct} is still correct,
whereas that in Figure~\ref{fig:mioco_example_incorrect} discards the mandatory action \emph{!cup} after the action \emph{?tea}, thus being incorrect regarding \mioco.

MIA-refinement is considered in two ways, first as an implementation relation (\miaref) and second as a relation for variant derivation (\confref).
For $\miocomia$ to yield a family-based conformance testing relation, it should be preserved by \confref, \ie if $i\miocomia s$ is checked for a product line implementation $i$ and its specification $s$, then this check can be neglected for the variants derivable from $i$.
Due to the fact that implementations are input-enabled, $\miocomia$ is also preserved by \miaref.
\begin{proposition}[MIA-Refinement preserves $\miocomia$]\label{theorem:mioco_refinement}
	Let $s, i$ be MIAs over $I$ and $O$ such that $i$ is input-enabled.
	If $i\miocomia s$, then for all $i'\miaref i$ it holds that $i'\miocomia s$.
\end{proposition}
\begin{proof}
The fact $i\miocomia s$ implies that $\mayout(i\mayafter \sigma) \subseteq \mayout(s\mayafter \sigma)$ for all $\sigma\in\maystraces(s)$ holds and $\mustout(s\mustafter \sigma) \subseteq \mustout(i\mustafter \sigma)$ for all $\sigma\in\muststraces(i)$ holds as well.
Let $i'$ be a MIA such that $i'\miaref i$.
Due to Lemma~\ref{lemma:refinement-outstates} and Lemma~\ref{lemma:input-enabledness}, $i'$ is input-enabled and for all $\sigma\in(I\cup O\cup\{\delta_\Diamond \})^{*}$, it holds that $\mayout(i'\mayafter \sigma)\subseteq \mayout(i\mayafter \sigma)$.
By transitivity of $\subseteq$, $\mayout(i'\mayafter \sigma)\subseteq\mayout(s\mayafter \sigma)$ holds for all $\sigma\in\maystraces(s)\subseteq(I\cup O\cup\{\delta_\Diamond \})^{*}$.

We now prove that also $\mustout(s\mustafter \sigma)\subseteq \mustout(i'\mustafter \sigma)$ for all $\sigma\in\muststraces(i')$.
As $i$ is input-enabled, it holds that $\mustout(i\mustafter\sigma)\subseteq\mustout(i'\mustafter\sigma)$ for all $\sigma\in(I\cup O\cup \{\delta_\Box\})^{*}$.
Therefore, $\mustout(s\mustafter\sigma)\subseteq\mustout(i'\mustafter\sigma)$ holds for all $\sigma\in\muststraces(i')\subseteq(I\cup O\cup \{\delta_\Box\})$ by transitivity of $\subseteq$.
Thus, $i'\miocomia s$.
\end{proof}
%
Next we show how to achieve input-enabledness by so-called \emph{completions}.
%
\subsection{Completions for MIA}
\label{subsec:completions}
In \mioco, we permit states to be underspecified, \ie we may leave open how a state $q\in Q$ of an implementation behaves in case of action $a\in(I\cup O)$ if $q\not\oversetmay{a}$.
Underspecification comes in two flavors: underspecification of input actions and underspecification of output actions.
Underspecification of output actions is explicit, \ie a state can only perform outputs attached to one of its transitions.
In contrast, underspecification of input actions is implicit, \ie a state accepts every possible input of the input set (even if there is no dedicated transition).
In this section, we present two transformations from underspecified to specified MIA accepting every possible input action, namely \emph{angelic completion} and \emph{chaotic completion}.
Both completions are described using the underspecified MIA depicted in Figure~\ref{fig:completions_example_specification}.
In this MIA with $I=\{\text{\emph{coffee, tea}}\}$ the two lower states are underspecified.
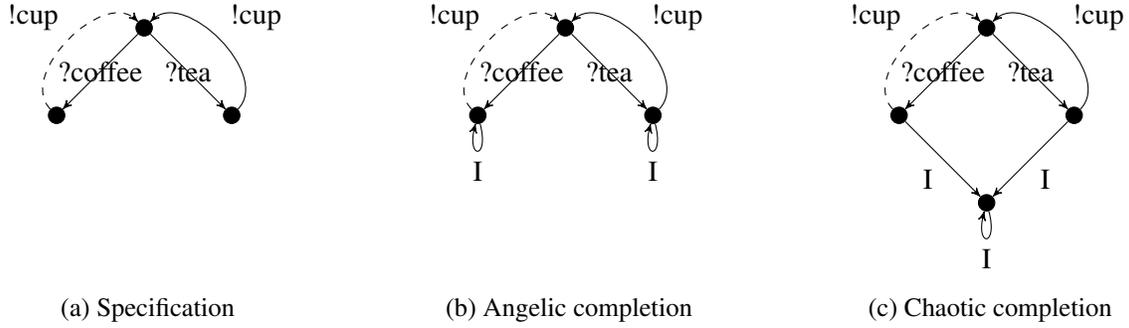
\begin{figure}
	\begin{subfigure}[b]{0.3\textwidth}
		\centering

\begin{tikzpicture}

\node[state] (q0) {};
\node[state, below left=of q0] (q1) {};
\node[state, below right=of q0] (q2) {};
\node[emptystate, below right=of q1] (q3) {};

\draw[must] (q0) to node[pos=.5] {?coffee} (q1);
\draw[may, bend left=90] (q1) to node[auto] {!cup} (q0);
\draw[must,bend right=90] (q2) to node[auto,swap] {!cup} (q0);
\draw[must] (q0) to node[pos=.5] {?tea} (q2);
\draw[emptyarrow,loop below] (q3) to node[auto] {\phantom{I}} (q3);

\end{tikzpicture}
		\subcaption{Specification}
		\label{fig:completions_example_specification}
	\end{subfigure}
	\hfill
	\begin{subfigure}[b]{0.3\textwidth}
		\centering

\begin{tikzpicture}

\node[state] (q0) {};
\node[state, below left=of q0] (q1) {};
\node[state, below right=of q0] (q2) {};
\node[emptystate, below right=of q1] (q3) {};

\draw[must] (q0) to node[pos=.5] {?coffee} (q1);
\draw[may, bend left=90] (q1) to node[auto] {!cup} (q0);
\draw[must,bend right=90] (q2) to node[auto,swap] {!cup} (q0);
\draw[must] (q0) to node[pos=.5] {?tea} (q2);
\draw[must,loop below] (q1) to node[auto] {I} (q1);
\draw[must,loop below] (q2) to node[auto] {I} (q2);
\draw[emptyarrow,loop below] (q3) to node[auto] {\phantom{I}} (q3);

\end{tikzpicture}
		\subcaption{Angelic completion}
		\label{fig:completions_example_angelic}
	\end{subfigure}
	\hfill
	\begin{subfigure}[b]{0.3\textwidth}
		\centering

\begin{tikzpicture}

\node[state] (q0) {};
\node[state, below left=of q0] (q1) {};
\node[state, below right=of q0] (q2) {};
\node[state, below right=of q1] (q3) {};

\draw[must] (q0) to node[pos=.5] {?coffee} (q1);
\draw[may, bend left=90] (q1) to node[auto] {!cup} (q0);
\draw[must,bend right=90] (q2) to node[auto,swap] {!cup} (q0);
\draw[must] (q0) to node[pos=.5] {?tea} (q2);
\draw[must] (q1) to node[auto,swap] {I} (q3);
\draw[must] (q2) to node[auto] {I} (q3);
\draw[must,loop below] (q3) to node[auto] {I} (q3);

\end{tikzpicture}
		\subcaption{Chaotic completion}
		\label{fig:completions_example_chaotic}
	\end{subfigure}
	\caption{Specification of a simplified vending machine and two completion strategies, where \emph{I} denotes one transition for both \emph{?coffee} and \emph{?tea}.}
	\label{fig:completions_example}
\end{figure}
One possibility for completion, called \emph{angelic} by Vaandrager~\cite{Vaandrager1991}, is to ignore unspecified inputs.
An angelically completed automaton \textit{MIA}$_{\textit{AC}}$ of a given \textit{MIA} is obtained by adding self-loop transitions to every state $q\in Q$ for every input $i\in I$ not being accepted by the state.
In Figure~\ref{fig:completions_example_angelic}, we added self-loops to the bottom states for input actions \emph{coffee} and \emph{tea}.
\begin{definition}[Angelic Completion]\label{def:angelic}
	Given a \textit{MIA} $(Q,I,O,\mustarrow,\mayarrow)$, its angelic completion \textit{MIA}$_{\textit{AC}}$ is defined as $(Q,I,O,\mustarrow',\mayarrow')$, where
	\begin{itemize}
		\item $\mustarrow'=\mustarrow\cup \{(q,i,q)\,|\,q\in Q\text{, }i\in I\text{, }q\not\oversetmust{i}\}$, and
		\item $\mayarrow'=\mayarrow\cup \{(q,i,q)\,|\,q\in Q\text{, }i\in I\text{, }q\not\oversetmay{i}\}$.
	\end{itemize}
\end{definition}
\emph{Chaotic completion} is also based on the work of Vaandrager~\cite{Vaandrager1991}, where the automaton is no more able to do any outputs as soon as an unspecified input occurred.
A chaotically completed automaton \textit{MIA}$_{CC}$ is obtained by adding a fresh error state which is entered whenever an unspecified input actions occurs.
In Figure~\ref{fig:completions_example_chaotic}, we added transitions from the states with underspecified input behavior to the error state.
Note that the error state is a so-called \emph{sink state}, because once reached, it will never be left.
\begin{definition}[Chaotic Completion]\label{def:chaotic}
	Given a \textit{MIA} $(Q,I,O,\mustarrow,\mayarrow)$, its chaotic completion \textit{MIA}$_{CC}$ is defined as $(Q',I,O,\mustarrow',\mayarrow')$, where
	\begin{itemize}
		\item $Q'=Q\cup\{q_E\}$, where $q_E\notin Q$,
		\item $\mustarrow'=\mustarrow\cup\{(q,i,q_E)\,|\,q\in Q\text{, }i\in I\text{, }q\not\oversetmust{i}\}\cup\{(q_E,\lambda,q_E)|\lambda\in I\}$, and
		\item $\mayarrow'=\mayarrow\cup\{(q,i,q_E)\,|\,q\in Q\text{, }i\in I\text{, }q\not\oversetmay{i}\}\cup\{(q_E,\lambda,q_E)|\lambda\in I\}$.
	\end{itemize}
\end{definition}
These results complete our discussions on modal testing theory based on MIA.
We now consider soundness and completeness notions for \mioco and give corresponding proofs.

\section{Correctness of \texorpdfstring{$\miocomia$}{mioco for MIA}}\label{sec:results} 
%
In order to serve as a reliable basis for family-based product line conformance testing, it is necessary for \mioco to be (1) sound, \ie whenever $i\miocomia s$ holds, then each implementation variant derivable from $i$ also conforms to a variant of $s$, and (2) complete, \ie whenever all variants of $i$ are correct w.\,r.\,t.\ \ioco and to the product line specification $s$, then $i\miocomia s$ holds.
In this section, we prove soundness of $\miocomia$ and discuss under which conditions completeness of \mioco may be obtained.
Whenever we use $\miocomia$ to show that a modal implementation $i$ conforms to the modal specification $s$, for each variant $i'\confref i$ there is a variant $s'\confref s$ such that $i'\mathioco s'$ holds.
By checking $i\miocomia s$ once, checking each and every variant of $i$ against some variant of $s$ can be omitted.
This is a remarkable improvement compared to variant-by-variant conformance testing, due to the exponentially growing number of variants $i'$ in the number of optional transitions of $i$.

For soundness, we need to take into account that all considered $i'$ are variants derived from $i$.
By Def.~\ref{def:conf_refinement}, each $i'$ contains at most those transitions being may transitions of $i$.
Therefore, $i$ restricts the set of possible transitions of $i'$ and as $i\miocomia s$ holds, also $s$ restricts the set of possible transitions of $s'$.
It is sufficient that the output behavior of $i'$ is included in that of $s'$, but not vice versa.
We, therefore, choose a single $s'$ for each $i'\confref i$, the {\em Family-LTS of $s$}, denoted by $\famlts s$, consisting of all may transitions of $s$.
%
	If $Q$ ($q$) is a MIA, then $\famlts{Q}$ ($\famlts{q}$) is the LTS with $\famlts{Q} = Q$ and $\famlts{\transrel{}}=\mayarrow$.
%
As $\famlts{\transrel{}}=\mayarrow$ and, therefore, $\famlts{\transrel{}}\subseteq\mayarrow$, it holds that $\famlts{q} \confref q$ for every MIA $q$.
Using $s' = \famlts{s}$ enables us to prove soundness.
\begin{theorem}[Soundness]\label{theorem:soundness}
  Let $s$ and $i$ be MIAs such that $i$ is input-enabled.
  If $i\miocomia s$, then for all $i'\confref i$, there exists some $s'\confref s$ such that $i'\mathioco s'$ holds.
\end{theorem}
%
%
\begin{proof}
  We prove $i\miocomia s\Rightarrow \forall i'\confref i: \exists s'\confref s: i'\mathioco s' $ by contradiction.
  We choose $s'$ to be $\famlts{s}$ for all $i'\confref i$.
  Assume that there is an $i'\confref i$ such that $i'\mathioco \famlts{s}$ does not hold, \ie there exists a $\sigma\in \straces(\famlts{s})$ such that $\out(i'\after\sigma)\not\subseteq \out(s'\after\sigma)$.
  By Lemma~\ref{lemma:lts_to_modal}, $\sigma\in\maystraces(s)$ and by construction of $\famlts{s}$ it holds that $\out(\famlts{s}\after\sigma) = \mayout(s\mayafter\sigma)$.
  It also holds that $\out(i'\after\sigma) \subseteq \mayout(i\mayafter\sigma)$, which implies that $\mayout(i\mayafter\sigma)\not\subseteq\mayout(s\mayafter\sigma)$ contradicting the assumption that $i\miocomia s$.
  Thus, there is no $i'\confref i$ such that $i'\mathioco \famlts{s}$ does not hold.
\end{proof}
%
%
The converse does not hold in general.
Consider the MIAs of Figure~\ref{fig:counterexample}, where $\famlts{s} = s$ and each variant $i'\confref i$ exhibits $i'\mathioco s$.
However, $i\miocomia s$ does not hold, as $s$ specifies an output $b$ as mandatory while in $i$, the $b$-transition is optional.
%
\begin{figure}
  \centering

\begin{tikzpicture}
	\node[state,label=above:$i$] (q0) {};
	\node[state, below left=of q0] (q1) {};
	\node[state, below right=of q0] (q2) {};

	\draw[must] (q0) to node[auto,swap] {!a} (q1);
	\draw[may] (q0) to node[auto] {!b} (q2);
\end{tikzpicture}
    \qquad\qquad
\begin{tikzpicture}
	\node[state,label=above:$s$] (q0) {};
	\node[state, below left=of q0] (q1) {};
	\node[state, below right=of q0] (q2) {};

	\draw[must] (q0) to node[auto,swap] {!a} (q1);
	\draw[must] (q0) to node[auto] {!b} (q2);
\end{tikzpicture}
  \caption{Each variant of $i$ conforms to $s$ (ioco), but $i\miocomia s$ does not hold.}\label{fig:counterexample}
\end{figure}
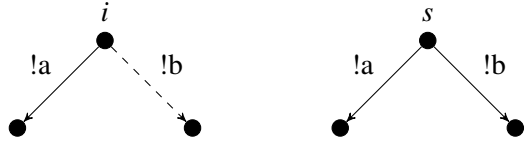
We observe that each \ioco check does not cover the fact that mandatory behavior of the specification $s$ must also be mandatory behavior of $i$.
This is due to the fact that in \ioco only allowed outputs may be implemented, but an obligation to implement any output, as imposed by must-modalities, is not covered.
If we ensure that mandatory behavior of $s$ is preserved by $i$, as \eg under MIA-refinement, the completeness claim holds.
Thus, we obtain the following completeness claim.
\begin{theorem}[Completeness I]\label{theorem:compl}
	Let $i, s$ be MIAs such that $i$ is input-enabled and $i\miaref s$.
	If for all $i'\confref i$ it holds that $i' \mathioco \famlts{s}$, then $i\miocomia s$.
\end{theorem}
\begin{proof}
  Assume $i\miocomia s$ does not hold, but for all $i'\confref i$ it holds that $i'\ioco \famlts{s}$.
  This means that (1) there exists a $\sigma\in\maystraces(s)$ such that $\mayout(i\mayafter\sigma)\not\subseteq\mayout(s\mayafter \sigma)$ or (2) there exists a $\sigma\in\muststraces(i)$ so that $\mustout(s\mustafter\sigma)\not\subseteq\mustout(i\mustafter\sigma)$.
  \begin{description}
    \item[Case (1):] It holds that $\sigma\in\maystraces(i)$.
    We construct a variant of $i$ respecting $\sigma$ as follows.
    Let $i'\miaref i$ the largest MIA (w.\,r.\,t.\ \miaref) such that whenever $i\Oversetmay{\sigma} q \oversetmay{a} q'$ then $i'\Oversetmust{\sigma} q\oversetmust{a} q'$.
    Hence, $\mustout(i'\mustafter\sigma) = \mayout(i\mayafter\sigma)$.
    $i_\sigma$ is the variant of $i$ that includes all must transitions of $i'$.
    But then $\out(i_\sigma\after\sigma) = \mustout(i'\mustafter\sigma)\not\subseteq\mustout(s\mustafter\sigma)$ and thus $i_\sigma\mathioco \famlts{s}$ does not hold, which contradicts the assumption that all variants of $i$ conform to $\famlts{s}$.
    \item[Case (2):] It holds that $\sigma\in\muststraces(s)$.
    As $\mustout(s\mustafter\sigma)\not\subseteq\mustout(i\mustafter\sigma)$, there is an $s'\in s\mustafter\sigma$ such that $s'\oversetmust{a}$ for some $a\in O$, but for all $i'\in i\mustafter\sigma$, it holds that $i'\not\oversetmust{a}$.
    But this contradicts the assumption that $i\miaref s$, as by Lemma~\ref{lemma:refinement-outstates} there is an $i'\in i\mustafter\sigma$ and $i'\miaref s'$.
  \end{description}
  Thus, $i'\mathioco \famlts{s}$ for all $i'\confref i$ implies that $i\miocomia s$.
\end{proof}
Thus, our \mioco framework is sound, and complete in case the implementation is a refined version of the specification.
When dropping the requirement of $i\miaref s$, it is possible to show that if there is a variant $i'$ of $i$ such that $i' \mathioco \famlts{s}$ does not hold, then $i\miocomia s$ does not hold, either.
\begin{theorem}[Completeness II]\label{theorem:contra-compl}
	Let $i, s$ be MIAs such that $i$ is input-enabled.
	If there is an $i'\confref i$ such that $i'\ioco \famlts{s}$ does not hold, then $i\miocomia s$ does not hold.
\end{theorem}
\begin{proof}
  Let $i'\confref i$ be an IOLTS such that $i'\mathioco \famlts{s}$ does not hold, \ie there exists a $\sigma\in\straces(\famlts{s})$ so that $\out(i'\after \sigma)\not\subseteq\out(\famlts{s}\after\sigma)$.
  By Lemma~\ref{lemma:lts_to_modal}, $\sigma\in\maystraces(s)$ and also $\mayout(i\mayafter\sigma) \neq\emptyset$.
  From the construction of $\famlts{s}$, $\mayout(i\mayafter\sigma)\not\subseteq\mayout(s\mayafter\sigma)$ implying $i\miocomia s$ does not hold.
\end{proof}
%
%
Theorem~\ref{theorem:soundness} ensures that whenever $\miocomia$ is established between product line implementation $i$ and product line specification $s$, then each variant $i'$ derived from $i$ I/O-conforms to $\famlts{s}$.
Correspondingly, Theorem~\ref{theorem:compl} and Theorem~\ref{theorem:contra-compl} state that whenever $\miocomia$ cannot be established between $i$ and $s$, then there is at least one variant $i'$ of $i$ not I/O-conforming to $\famlts{s}$.
According to Theorem~\ref{theorem:compl}, this is only ensured if $i\miaref s$ holds.
Summarizing, our $\miocomia$ reflects the essence of family-based product line analysis~\cite{Thum2014} by means of I/O-conformance testing.

\section{Conclusion and Future Work}\label{sec:discussion}
In this paper, we proposed a family-based I/O-conformance testing theory for product lines based on Modal Interface Automata, which is sound and complete w.\,r.\,t.\ variant-by-variant I/O-conformance testing based on IOLTS.
As future work, we plan to exploit the MIA framework for its compositionality properties to obtain criteria for compositional I/O-conformance testing of product lines.
Therefore, dealing with internal actions, excluded from this papers' considerations, is inevitable.
However, the results we obtained throughout this paper canonically extend to the case of MIAs with internal actions.
This way, we obtain a similar variability concept as Larsen et al.~\cite{larsen:modalioautomata}, which is based on modal refinement and the ability of composition with an environmental specification validating implementation variants.
Furthermore, we plan to implement our theory, based on a $\miocomia$-extended version of {\scshape JTorX}~\cite{Belinfante2010} to provide an applicable tool for efficient product line I/O conformance testing.
    \bibliographystyle{eptcs}
    \bibliography{content/sources.eptcs}

\end{document}